\documentclass[12pt,reqno]{amsart}

\usepackage{vmargin,fancyhdr}   
\usepackage{amsmath}
\usepackage{amsfonts}
\usepackage{hyperref}
\usepackage{amssymb}
\usepackage{color}
\usepackage{graphicx}
\usepackage{subfig}
\usepackage{graphicx}
\usepackage{times}

\newcommand{\whp}{\textit{whp}~}

\newcommand{\hide}[1]{}

\newcommand{\field}[1]{\mathbb{#1}} 

\newcommand{\beql}[1]{\begin{equation}\label{#1}}
\newcommand{\eeq}{\end{equation}}
\newcommand{\comment}[1]{}

\newcommand{\Mean}[1]{{\mathbb E}\left[{#1}\right]}

\newcommand{\Ceil}[1]{{\left\lceil{#1}\right\rceil}}

\newcommand{\Prob}[1]{{{\bf{Pr}}\left[{#1}\right]}}

\newtheorem{theorem}{Theorem}
\newtheorem{lemma}{Lemma}

\newtheorem{claim}{Claim}

\begin{document}

\title{High Degree Vertices, Eigenvalues and Diameter of Random Apollonian Networks}

\author{Alan Frieze}
\address{Department of Mathematical Sciences\\
Carnegie Mellon University\\
5000 Forbes Av., 15213\\
Pittsburgh, PA \\
U.S.A} \email{alan@random.math.cmu.edu}
\author{Charalampos E. Tsourakakis}
\address{Department of Mathematical Sciences\\
Carnegie Mellon University\\
5000 Forbes Av., 15213\\
Pittsburgh, PA \\
U.S.A} \email{ctsourak@math.cmu.edu}

\maketitle

\begin{abstract}

Upon the discovery of power laws \cite{albert,broder,faloutsos},
a large body of work in complex network analysis has focused on developing 
generative models of graphs which mimick real-world network properties
such as skewed degree distributions \cite{faloutsos}, small diameter \cite{albert2} and large clustering
coefficients \cite{tsourakakiswaw,strogatz}.  Most of these models belong either to the stochastic, 
e.g., \cite{albert,bonato,cooperfrieze,kronecker},
or the strategic e.g., \cite{aumann,aumann2,boorman,fabrikant}, family of network formation models. 

Despite the fact that planar graphs arise in numerous real-world settings, e.g., in road and railway maps, 
in printed circuits, in chemical molecules, in river networks  \cite{urban,london}, 
comparably less attention has been devoted to the study of planar graph generators. 
In this work we analyze basic properties of Random Apollonian Networks \cite{zhang,zhou}, 
a popular stochastic model which generates planar graphs with power law properties. 

Specifically, let $k$ be a constant and $\Delta_1 \geq \Delta_2 \geq .. \geq \Delta_k$ be the
degrees of the $k$ highest degree vertices. We prove that at time $t$, for any function 
$f$ with $f(t) \rightarrow +\infty$  as $t \rightarrow +\infty$, $\frac{t^{1/2}}{f(t)} \leq \Delta_1 \leq f(t)t^{1/2}$
and for $i=2,\ldots,k=O(1)$,  $\frac{t^{1/2}}{f(t)} \leq \Delta_i \leq \Delta_{i-1} - \frac{t^{1/2}}{f(t)}$
with high probability (\whp). Then, we show that the $k$ largest eigenvalues of the adjacency
matrix of this graph satisfy $\lambda_k = (1\pm o(1))\Delta_k^{1/2}$ \whp.
Furthermore, we prove a refined upper bound on the asymptotic growth of the diameter, i.e., 
that \whp the diameter $d(G_t)$ at time $t$ satisfies $d(G_t) \leq \rho \log{t}$
where $\frac{1}{\rho}=\eta$ is the unique solution greater than 1 of the equation $\eta - 1 - \log{\eta} = \log{3}$.
Finally, we investigate other properties of the model.

\end{abstract}

\section{Introduction} 
\label{sec:intro} 

In recent years, a considerable amount of research has focused on the study of graph structures arising from technological,
biological and sociological systems. Graphs are the tool of choice in modeling such systems since the latter are typically described as a set of pairwise interactions. Important examples of such datasets are the Internet graph (vertices are routers, edges correspond to physical links), the Web graph (vertices are web pages, edges correspond to hyperlinks), social networks (vertices are humans, edges correspond to friendships), information networks like Facebook and LinkedIn (vertices are accounts, edges correspond to online friendships), biological networks (vertices are proteins, edges correspond to protein interactions), math collaboration network (vertices are mathematicians, edges correspond to collaborations) and many more.

Towards the end of the '90s, a series of papers observed that the classic models of random graphs introduced by Erd\"{o}s and R\'{e}nyi \cite{erdos,erdos2}
and Gilbert \cite{gilbert} did not explain the empirical properties of real-world networks \cite{albert,broder,faloutsos}. 
Typical properties of such networks include skewed degree distributions \cite{faloutsos}, 
large clustering coefficients \cite{tsourakakiswaw,strogatz} and small average distances \cite{albert2,strogatz}, a phenomenon 
typically referred to as ``small worlds''.  Skewed degree distributions have widely been modeled as power laws:
the proportion of vertices of a given degree follows an approximate inverse power law, i.e., the proportion of vertices of
degree k is approximately $Ck^{-\alpha}$.

Understanding the properties of real world networks  has attracted considerable research interest in the recent years \cite{chunglu}.
A large body of work has focused on finding models which generate graphs which mimick real world networks. 
Most of the existing work on network formation models falls into the stochastic and the strategic categories.
Kronecker graphs \cite{kronecker}, the Cooper-Frieze model for the Web graph \cite{cooperfrieze}, 
the Aiello-Chung-Lu model \cite{aiello}, Protean graphs \cite{protean} and numerous other models 
belong to the former category. 
The strategic approach has its origins in the work of Boorman \cite{boorman}, Aumann \cite{aumann}
and Aumann and Myerson \cite{aumann2} and a variety of models fall into this category, 
e.g., the Fabrikant-Koutsoupias-Papadimitriou model \cite{fabrikant}.
Despite the large amount of work on models which generate power law graphs, 
a considerably smaller amount of work has focused on generative models for planar graphs.
Planar graphs have been studied mainly in the context of transportation networks
and occur in numerous real-world problems: modelling
city streets \cite{eisenstat,london}, crowd simulations \cite{pedestrians},
river networks, railway and road maps printed circuits, chemical molecules,
see also  \cite{urban} and references therein. 

\noindent In this work, we prove fundamental properties of Random Apollonian Networks (RANs) \cite{zhou},
a popular model of planar graphs with power law properties \cite{survey}. We use the symbol 
$G_t$ to denote the random graph at time $t$. The details of the model are exposed in 
Section~\ref{sec:ran}. Specifically our main results are the following theorems:

\begin{theorem}[Highest $k$ Degrees]
Let $\Delta_1 \geq \Delta_2 \geq \ldots \geq \Delta_k$ be the $k$ highest degrees of the Random Apollonian Network $G_t$ 
at time $t$ where $k$ is a fixed positive integer. Also, let $f(t)$ be a function such that $f(t) \rightarrow +\infty$ as $t \rightarrow +\infty$. Then \whp

$$ \frac{t^{1/2}}{f(t)} \leq \Delta_1 \leq t^{1/2} f(t) $$ 

\noindent and for $i=2,\ldots,k$

$$ \frac{t^{1/2}}{f(t)} \leq \Delta_i \leq \Delta_{i-1} - \frac{t^{1/2}}{f(t)} $$ 
\label{thrm:thrm1} 
\end{theorem}

\noindent The growing function $f(t)$ cannot be removed, see \cite{flaxman}.
Using Theorem~\ref{thrm:thrm1} we relate the highest $k$ degrees and eigevalues. 

\begin{theorem}[Largest $k$ Eigenvalues] 
Let $\lambda_1 \geq \lambda_2 \geq \ldots \geq \lambda_k$ be the largest $k$ eigenvalues
of the adjacency matrix of $G_t$. Then \whp $ \lambda_i =(1\pm o(1))\sqrt{\Delta_i}.$
\label{thrm:thrm2} 
\end{theorem}

\noindent Also, we prove the following theorem for the diameter. 

\begin{theorem}[Diameter] 
The diameter $d(G_t)$ of $G_t$ satisfies in probability $ d(G_t) \leq \rho \log{t}$
where $\frac{1}{\rho}=\eta$ is the unique solution greater than 1 of the equation $\eta - 1 - \log{\eta} = \log{3}$.
\label{thrm:ternary} 
\end{theorem}

The outline of the paper is as follows: in Section~\ref{sec:ran} we describe the model and discuss
existing work on RANs. In Sections~\ref{sec:degree} and~\ref{sec:eigen}  we prove Theorems~\ref{thrm:thrm1}
and~\ref{thrm:thrm2} respectively. 
In Section~\ref{sec:diameter}  we give a simple proof for the asymptotic growth of the diameter
and prove Theorem~\ref{thrm:ternary}. Finally in Section~\ref{sec:waiting} 
we investigate another property of the model.
Finally, in Section~\ref{sec:concl} we conclude by proposing several problems for future work.

\section{Random Apollonian Networks}
\label{sec:ran} 

\begin{figure}
  \centering
  \subfloat[$t=1$]{\label{fig:fig1a}\includegraphics[width=0.2\textwidth]{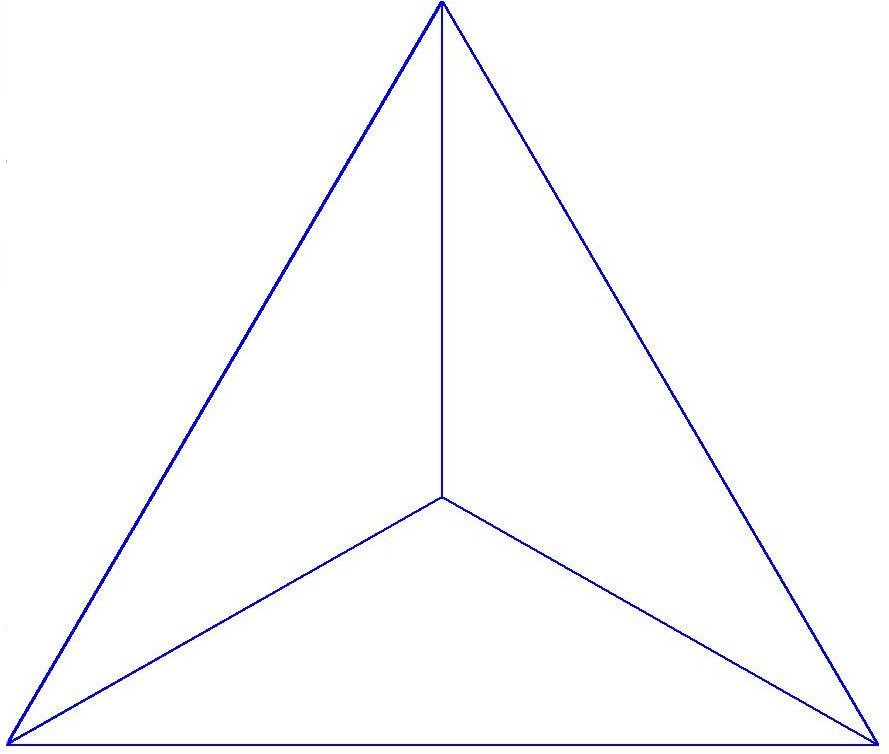}}                
  \subfloat[$t=2$]{\label{fig:fig1b}\includegraphics[width=0.2\textwidth]{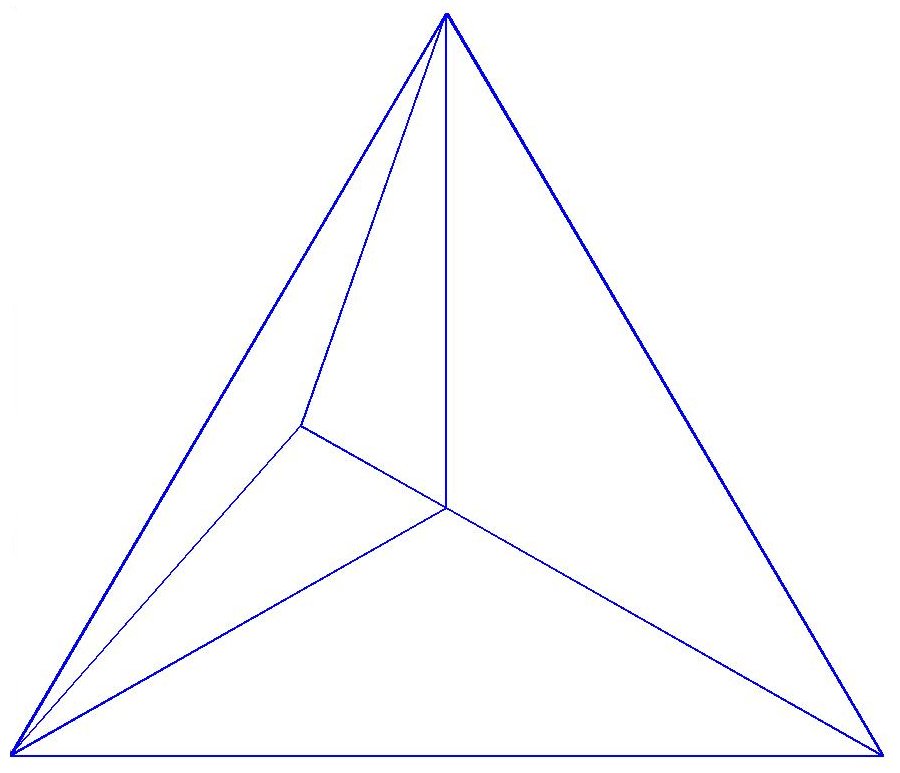}}
  \subfloat[$t=3$]{\label{fig:fig1c}\includegraphics[width=0.2\textwidth]{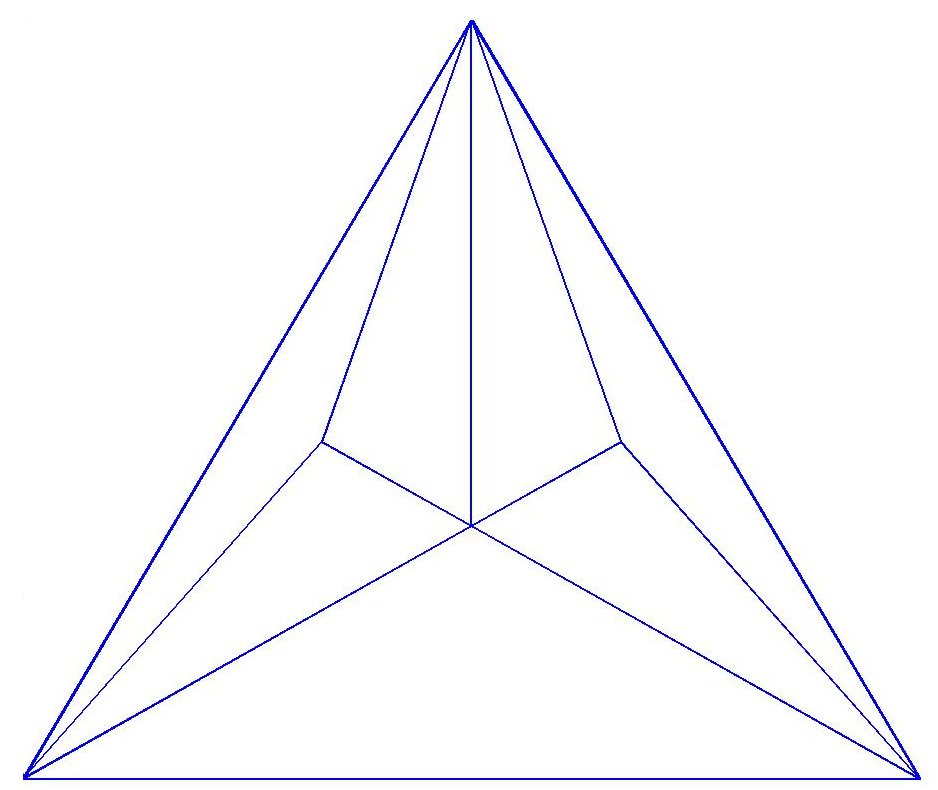}}
  \subfloat[$t=100$]{\label{fig:fig1d}\includegraphics[width=0.2\textwidth]{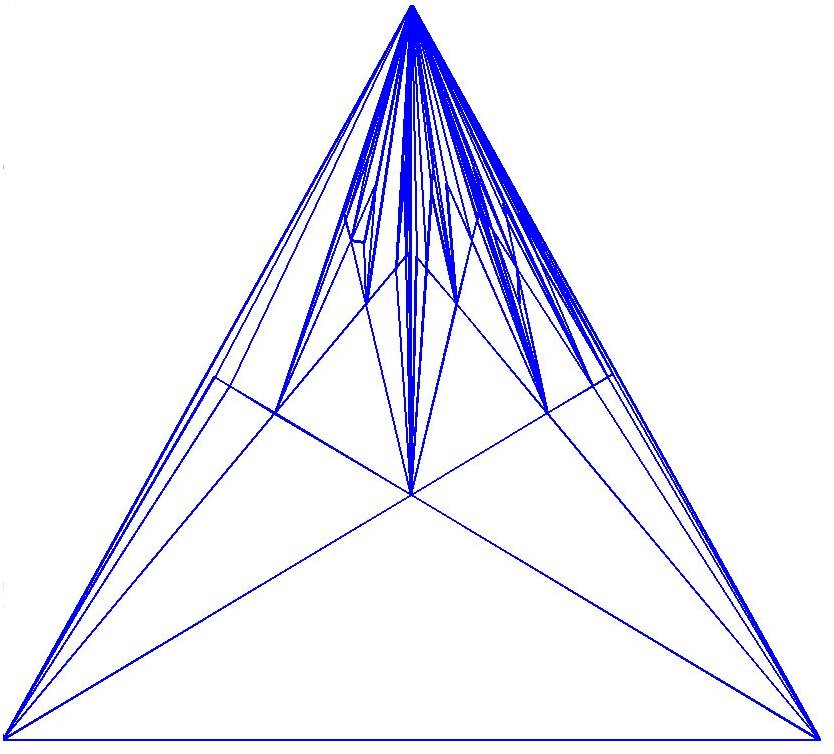}} 
  \caption{Snapshots of a Random Apollonian Network (RAN) at: (a) $t=1$ (b) $t=2$ (c) $t=3$ (d) $t=100$. }
  \label{fig:fig1}
\end{figure}

Apollonius of Perga was a Greek geometer and astronomer noted for his writings on conic sections.
Apollonius introduced the problem of space filling packing of spheres which in its classical solution, 
the Apollonian packing --see \cite{graham} for remarkable properties of Apollonian circles-, exhibits a power law behavior. Specifically,  
the circle size distribution follows a power law with exponent of about 1.3 \cite{boyd}.
Apollonian Networks (ANs)  were introduced in \cite{adrade} and independently in \cite{doye}. 
Zhou et al. \cite{zhou} introduced the Random Apollonian Networks (RANs) and  gave formulae for their order, size, degree distribution (using heuristic arguments), clustering coefficient and diameter. 
High dimensional RANs were introduced in \cite{zhang}. 
The degree distribution of RANs was shown to follow a power law in \cite{comment,zhou}. 
Other properties have also been analyzed, e.g., average distance of two vertices \cite{darase2}, properties of the connectivity profile \cite{darase3}.  Zhou et al. \cite{zhou} proposed a simple rule that generates
a random two-dimensional Apollonian networks with very large clustering coefficient.
RANs are planar 3-trees, a special case of random $k$-trees \cite{kloks}. 
The general result of the degree distribution of random $k$-trees was proved by Cooper \& Uehara and Gao \cite{cooper,gao}. 
In RANs --in contrast to the general model of random $k$ trees-- the random $k$ clique 
chosen at each step has never previously been selected. For example in the two dimensional 
case any chosen triangular face is being subdivided into three new triangular faces by connecting
the incoming vertex to the vertices of the boundary. 
Random $k$-trees due to their power law properties have been proposed as a 
model for complex networks, see, e.g., \cite{cooper,gao2} and references therein.
Recently, a variant of $k$-trees, namely ordered increasing $k$-trees has been proposed and analyzed in \cite{panholzer}.

An example of a two dimensional RAN is shown in Figure~\ref{fig:fig1}. The RAN generator takes as input a parameter $T_{\max}$ equal to the number of iterations the algorithm will perform and runs as follows: 

\begin{itemize}
 \item Let $G_0$ be a triangle $(1,2,3)$ embedded on the plane, e.g., as an equilateral triangle.
 \item {\it for} $t \leftarrow 1$ to $T_{\max}$: 
 \begin{itemize}
  \item  Sample a face $F_t=(i,j,k)$ of the planar  graph $G_{t-1}$ uniformly at random.
  \item  Insert the vertex $t+3$ inside this face (e.g., in the barycenter of the corresponding triangle) and
   draw the three edges $(i,t+3)$, $(j,t+3)$,$(k,t+3)$, e.g., as straight lines.
 \end{itemize}
\end{itemize}

It's worth pointing out why we expect this model to yield a power law degree distribution:
consider any vertex $v$ at some time $t$. Vertex $v$ sees a wheel graph around it (except for vertices 
1,2,3 who see a wheel modulo one edge) as Figure~\ref{fig:fig2} shows. Since we pick a face uniformly
at random the probability that $v$ increases its degree by 1 is (roughly) proportional to its degree. 
Therefore, we expect this process to result in a graph, which is obviously planar, and which \whp
\footnote{A sequence of events $\mathcal{E}_n$ occurs with high probability \whp if $\lim_{n \rightarrow +\infty} \Prob{\mathcal{E}_n}=1$.}
has properties similar to those generated by a preferential attachment process \cite{albert}. 
Therefore, it should not come  as a surprise that existing techniques developed in the context of
preferential attachment models can be adapted to solve problems concerning the structure 
of Random Apollonian Networks. 

Finally, we shall make use of the following formulae for the number of vertices ($n_t$), edges ($m_t$) and faces ($F_t$) at time $t$ in a RAN $G_t$:
$$n_t=t+3,~~ m_t=3t+3,~~ F_t=2t+1.$$ 
\noindent Note that a RAN is a maximal planar graph since for any planar graph $m_t \leq 3n_t - 6 \leq 3t+3$.

\begin{figure}
  \centering
   {\label{fig:fig2a}\includegraphics[width=0.25\textwidth]{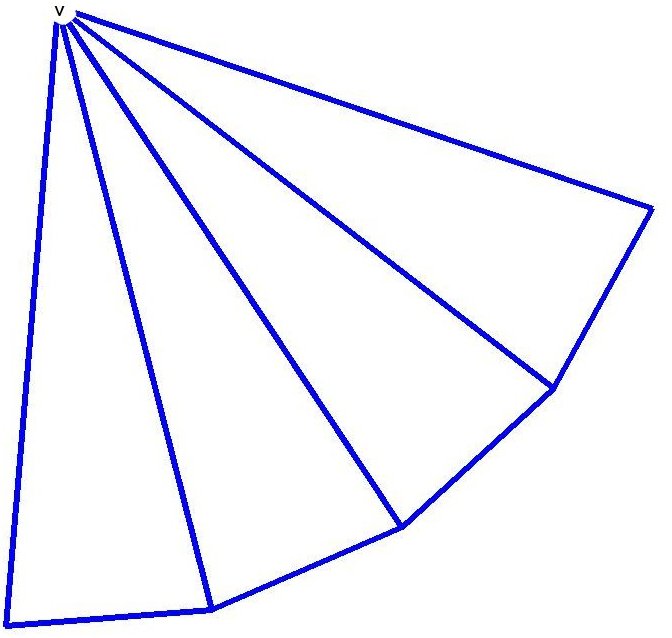}}                
   {\label{fig:fig2b}\includegraphics[width=0.25\textwidth]{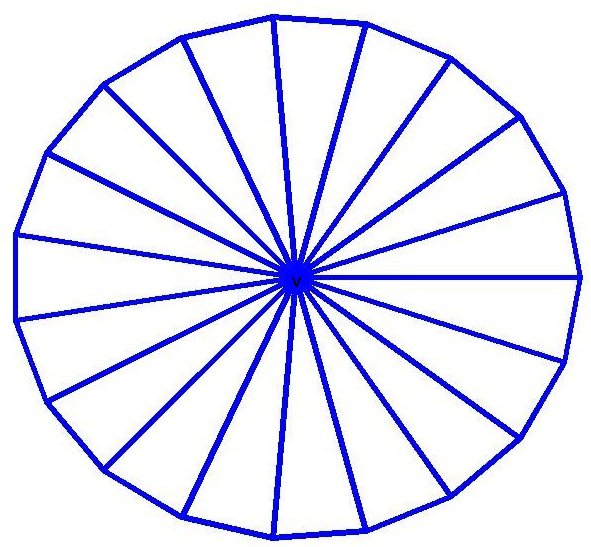}}
  \caption{What each vertex $v$ sees at time $t$: (a) Vertices $1,2,3$ (b) The other vertices $ \{4,\ldots,t+3\}$. 
  Since we pick a face uniformly at random the probability that the degree of $v$ increases by one is (roughly) proportional to its degree. 
  Therefore, we expect this process to generate a planar graph with properties similar to those generated by a preferential attachment    
  process. }
  \label{fig:fig2}
\end{figure}

\section{Highest Degree Vertices: Proof of Theorem~\ref{thrm:thrm1}} 
\label{sec:degree} 

We decompose the proof of Theorem~\ref{thrm:thrm1} into several lemmas
which we prove in the following. The proof of Theorem~\ref{thrm:thrm1}
follows from Lemmas~\ref{lem:lemma2},~\ref{lem:lemma3},~\ref{lem:lemma4},~\ref{lem:lemma5},~\ref{lem:lemma6}.
We partition the vertices into three sets: those added before $t_0$, before
$t_1$ and after $t_1$ where $t_0=\log{\log{\log{(f(t))}}}$ and $t_1=\log{\log{(f(t))}}$.
We define a supernode to be a collection of vertices and the degree of the supernode the sum of the degrees 
of its vertices. 

\begin{lemma} 
Let $d_t(s)$ denote the degree of vertex $s$ at time $t$. and let $a^{(k)}=a(a+1)\ldots (a+k-1)$\
denote the rising factorial function. Then, for any positive integer $k$

\begin{equation}
\Mean{ d_t(s)^{(k)} } \leq \frac{ (k+2)!}{2} \big( \frac{2t}{s} \big)^{\frac{k}{2}}.
\label{eq:rising} 
\end{equation}

\label{lem:lemma1}
\end{lemma}

\begin{proof} 
We distinguish two types of vertices\footnote{Despite the fact that 
the results don't change asymptotically, we treat both cases in Lemma~\ref{lem:lemma1}. 
This analysis will be omitted in the subsequent lemmas.}
since --as it can also be seen in Figure~\ref{fig:fig2}-- the three initial vertices $1,2,3$
have  one face less than their degree whereas all other vertices have degree equal to the number
of faces.

\noindent \underline{$\bullet$ {\sc Case 1} $s \geq 4$:}\\

\noindent  Note that $d_s(s)=3$ and that $ \frac{2t}{s} \geq \frac{2t-1}{2s-1}$. By conditioning successively we obtain 

\begin{align*} 
\Mean{d_t(s)^{(k)}}                &= \Mean{ \Mean{ d_t(s)^{(k)} | d_{t-1}(s)} }\\
                                   &= \Mean{ (d_{t-1}(s))^{(k)} \big( 1-\frac{d_{t-1}(s)}{2t-1} \big) +  (d_{t-1}(s)+1)^{(k)} \frac{d_{t-1}(s)}{2t-1}} \\ 
                                   &=  \Mean{  (d_{t-1}(s))^{(k)} \big( 1-\frac{d_{t-1}(s)}{2t-1} \big) + (d_{t-1}(s))^{(k)} \frac{d_{t-1}(s)+k}{d_{t-1}(s)} \frac{d_{t-1}(s)}{2t-1}}\\
                                   &=  \Mean{ (d_{t-1}(s))^{(k)}}\big( 1+\frac{k}{2t-1} \big) = ... = 3^{(k)} \prod_{t'=s+1}^t ( 1+\frac{k}{2t'-1} ) \\
                                   &\leq  3^{(k)}  \exp{\Big( \sum_{t'=s+1}^t \frac{k}{2t'-1} \Big)} \leq 3^{(k)} \exp{ \Big( k \int_{s}^{t} \! \frac{\mathrm{d}x}{2x-1}   \,   \Big) }  \\ 
                                   &\leq \frac{(k+2)!}{2} \exp{ \Big( k \log{ \frac{2t-1}{2s-1} } \Big)} \leq \frac{ (k+2)!}{2} \Big( \frac{2t}{s} \Big)^{\frac{k}{2}}. 
\end{align*}

\noindent \underline{$\bullet$ {\sc Case 2} $s \in \{1,2,3\}$:}\\

\noindent  Note that initially the degree of any such vertex is 2. For any $k\geq 0$

\begin{align*} 
\Mean{d_t(s)^{(k)}}                &= \Mean{ \Mean{ d_t(s)^{(k)} | d_{t-1}(s)} }\\
                                   &= \Mean{ (d_{t-1}(s))^{(k)} \big( 1-\frac{d_{t-1}(s)-1}{2t-1} \big) +  (d_{t-1}(s)+1)^{(k)} \frac{d_{t-1}(s)-1}{2t-1} }\\ 
                                   &=  \Mean{ (d_{t-1}(s))^{(k)} \big( 1+\frac{k}{2t-1} \big) - (d_{t-1}(s))^{(k)} \frac{k}{(2t-1)d_{t-1}(s)}  }\\
                                   &\leq \Mean{ (d_{t-1}(s))^{(k)}} \big( 1+\frac{k}{2t-1} \big) \leq \ldots \leq \frac{ (k+2)!}{2} \big( \frac{2t}{s} \big)^{\frac{k}{2}}.
\end{align*}

\end{proof}

\begin{lemma}
The degree $X_t$ of the supernode $V_{t_0}$ of vertices added before time $t_0$ is at least $t_0^{1/4}\sqrt{t}$ \whp.
\label{lem:lemma2}
\end{lemma}

\begin{proof} 

We consider a modified process $\mathcal{Y}$ coupled with the RAN process, see also Figure~\ref{fig:fig3}.
Specifically, let $Y_t$ be the modified degree of the supernode in the modified process $\mathcal{Y}$ which 
is defined as follows: for any type of insertion in the original RAN process --note there exist three types of insertions with respect 
to how the degree $X_t$ of the supernode (black circle) gets affected, see also Figure~\ref{fig:fig3}-- 
$Y_t$ increases by 1. We also define $X_{t_0}=Y_{t_0}$.  Note that $X_t \geq Y_t$ for all $t \geq t_0$.
Let $d_0=X_{t_0}=Y_{t_0}=6t_0+6$ and  $p^*=\Prob{Y_t=d_0+r| Y_{t_0}=d_0}$.

\begin{figure*}
\centering
\includegraphics[width=0.5\textwidth]{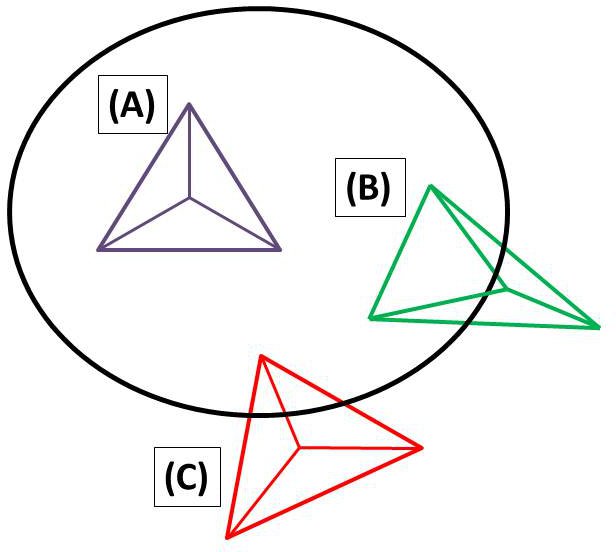}
\caption{Coupling of Lemma~\ref{lem:lemma2}.}
\label{fig:fig3}
\end{figure*}

\noindent The following technical claim is proved in the Appendix~\ref{sec:appendix}.
\begin{claim}
$$  p^* \leq  {d_0+r-1 \choose d_0-1} \Big(\frac{2t_0+3}{2t+1}\Big)^{d_0/2} e^{\frac{3}{2}+t_0-\frac{d_0}{2}+\frac{2r}{3\sqrt{t}}} $$
\label{claim:lemma2}
\end{claim}

Let $\mathcal{A}_1$ denote the event that the supernode consisting of the first $t_0$ vertices has degree $Y_t$
in the modified process $\mathcal{Y}$ less than $t_0^{1/4}\sqrt{t}$. 
Note that since $\{ X_t \leq t_0^{1/4} \sqrt{t}\} \subseteq \{ Y_t \leq t_0^{1/4}\sqrt{t}\} $ it 
suffices to prove that $\Prob{Y_t \leq t_0^{1/4}\sqrt{t}} =  o(1)$. Using Claim~\ref{claim:lemma2} we obtain

\begin{align*} 
\Prob{\mathcal{A}_1} &\leq \sum_{r=0}^{ t_0^{1/4}\sqrt{t} -(6t_0+6)} {r+6t_0+5 \choose 6t_0+5}  \Big(\frac{2t_0+3}{2t+1}\Big)^{3t_0+3} e^{-\frac{3}{2}-2t_0+\frac{2t_0^{1/4}}{3}} \\
&\leq t_0^{1/4}t^{1/2} \frac{\big( t_0^{1/4}t^{1/2} \big)^{6t_0+5}}{(6t_0+5)!} \Big(\frac{2t_0+3}{2t+1}\Big)^{3t_0+3}
e^{-\frac{3}{2}-2t_0+\frac{2t_0^{1/4}}{3}} \\ 
&\leq \Big( \frac{t}{2t+1} \Big)^{3t_0+3}  \frac{t_0^{3t_0/2+3/2} (2t_0+3)^{3t_0+3}}{(6t_0+5)^{6t_0+5}} e^{4t_0+7/2+2/3 t_0^{1/4}}\\
&\leq 2^{-(3t_0+3)} \frac{e^{4t_0+7/2+2/3 t_0^{1/4}}}{(6t_0+5)^{\tfrac{3}{2}t_0+\tfrac{1}{2}}}  =o(1).
\end{align*}

\end{proof}

\begin{lemma}
No vertex added after $t_1$ has degree exceeding $t_0^{-2}t^{1/2}$ \whp.
\label{lem:lemma3}
\end{lemma}

\begin{proof} 

Let $\mathcal{A}_2$ denote the event that some vertex added after $t_1$ has
degree exceeding $t_0^{-2}t^{1/2}$. We use a union bound, a third moment argument
and Lemma~\ref{lem:lemma1} to prove that $\Prob{\mathcal{A}_2}=o(1)$. 
Specifically 

\begin{align*}
\Prob{\mathcal{A}_2} &\leq  \sum_{s=t_1}^t \Prob{d_t(s) \geq t_0^{-2}t^{1/2}} = \sum_{s=t_1}^t \Prob{d_t(s)^{(3)} \geq (t_0^{-2}t^{1/2})^{(3)}} \\ 
                     &\leq  t_0^{6}t^{-3/2} \sum_{s=t_1}^t \Mean{d_t(s)^{(3)}} \leq 5!\sqrt{2} t_0^6 \sum_{s=t_1}^t s^{-3/2}            \leq 5!2\sqrt{2} t_0^6 t_1^{-1/2} =o(1). 
\end{align*} 

\end{proof}

\begin{lemma}
No vertex added before $t_1$ has degree exceeding $t_0^{1/6}t^{1/2}$ \whp. 
\label{lem:lemma4}
\end{lemma}

\begin{proof} 
Let $\mathcal{A}_3$ denote the event that some vertex added before $t_1$ 
has degree exceeding $t_0^{1/6}t^{1/2}$. We use again a third moment argument
and Lemma~\ref{lem:lemma1} to prove that  $\Prob{\mathcal{A}_3}=o(1)$.

\begin{align*} 
\Prob{\mathcal{A}_3} &\leq \sum_{s=1}^{t_1} \Prob{ d_t(s) \geq t_0^{1/6}t^{1/2}} = \sum_{s=1}^{t_1} \Prob{ d_t(s)^{(3)} \geq (t_0^{1/6}t^{1/2})^{(3)} } \\ 
           &\leq t_0^{-1/2}t^{-3/2} \sum_{s=1}^{t_1} \Mean{ d_t(s)^{(3)} }  \leq t_0^{-1/2}t^{-3/2} \sum_{s=1}^{t_1} 5!\sqrt{2} \frac{t^{3/2}}{s^{3/2}} \\ 
           &\leq 5!\sqrt{2}\zeta(3/2)t_0^{-1/2} =o(1)
\end{align*} 

\noindent where $\zeta(3/2)=\sum_{s=1}^{+\infty} s^{-3/2} \approx 2.612$.

\end{proof}

\begin{lemma}
The $k$ highest degrees are added before $t_1$ and have degree $\Delta_i$ bounded by 
$ t_0^{-1} t^{1/2} \leq \Delta_i \leq t_0^{1/6}t^{1/2}$  \whp.
\label{lem:lemma5}
\end{lemma}

\begin{proof} 
For the upper bound it suffices to show that $\Delta_1 \leq t_0^{1/6}t^{1/2}$. 
This follows immediately by Lemmas~\ref{lem:lemma3} and~\ref{lem:lemma4}.
The lower bound follows directly from Lemmas~\ref{lem:lemma2},~\ref{lem:lemma3} and~\ref{lem:lemma4}.
Assume that at most $k-1$ vertices added before $t_1$ have degree exceeding the lower bound 
$t_0^{-1} t^{1/2}$. Then the total degree of the supernode formed by the first $t_0$ vertices
is $O(t_0^{1/6}\sqrt{t})$. This contradicts Lemma~\ref{lem:lemma2}. 
Finally, since each vertex $s\geq t_1$ has degree at most $t_0^{-2}\sqrt{t} \ll t_0^{-1} t^{1/2}$
the $k$ highest degree vertices are added before $t_1$ \whp.
\end{proof}

\begin{lemma}
The $k$ highest degrees satisfy $ \Delta_{i} \leq \Delta_{i-1} - \frac{\sqrt{t}}{f(t)}$ \whp.
\label{lem:lemma6}
\end{lemma}

\begin{proof} 
Let $\mathcal{A}_4$ denote the event that there are two vertices among the first $t_1$
with degree $t_0^{-1} t^{1/2}$ and within $\frac{\sqrt{t}}{f(t)}$ of each other. 
By the definition of conditional probability and Lemma~\ref{lem:lemma3}

\begin{align*}
\Prob{\mathcal{A}_4}  &= \Prob{\mathcal{A}_4|\bar{\mathcal{A}_3}} \Prob{\bar{\mathcal{A}_3}} + \Prob{\mathcal{A}_4|\mathcal{A}_3}\Prob{\mathcal{A}_3} \leq \Prob{\mathcal{A}_4|\bar{\mathcal{A}_3}} + o(1) 
\end{align*} 

\noindent it suffices to show that $\Prob{\mathcal{A}_4 | \bar{\mathcal{A}_3}}=o(1)$. Note that by a simple union bound

\begin{align*}
\Prob{\mathcal{A}_4} &\leq \sum_{1 \leq s_1 < s_2 \leq t_1} \sum_{l=-\frac{\sqrt{t}}{f(t)}}^{\frac{\sqrt{t}}{f(t)}} p_{l,s_1,s_2} =O\Big(t_1^2 \frac{\sqrt{t}}{f(t)} \max{ p_{l,s_1,s_2} }\Big)
\end{align*}

\noindent where $p_{l,s_1,s_2} = \Prob{d_t(s_1)-d_t(s_2)=l|\bar{\mathcal{A}_3}}$. 

We consider two cases and we show that in both cases $\max{ p_{l,s_1,s_2}}=o(\frac{f(t)}{t_1^2\sqrt{t}})$.

\noindent
\underline{$\bullet$ {\sc Case 1} $(s_1,s_2) \notin E(G_t)$:}\\
 
Note that at time $t_1$ there exist $m_{t_1}=3t_1+3 < 4t_1$ edges in $G_{t_1}$. 

\begin{align*} 
p_{l,s_1,s_2} &\leq \sum_{r=t_0^{-1}t^{1/2}}^{t_0^{1/6}t^{1/2}} \sum_{d_1,d_2=3}^{4t_1} \Prob{d_t(s_1)=r \wedge d_t(s_2)=r-l | d_{t_1}(s_1)=d_1, d_{t_1}(s_2)=d_2} \tag{2}\\ 
              &\leq t_0^{1/6}t^{1/2}  \sum_{d_1,d_2=3}^{4t_1} {2t_0^{1/6}t^{1/2} \choose d_1-1} {2t_0^{1/6}t^{1/2} \choose d_2-1} \Big(\frac{2t_0+3}{2t+1}\Big)^{(d_1+d_2)/2} e^{\frac{3}{2}+t_1+\frac{2t_0^{1/6}}{3}} \tag{3}\\ 
              &\leq t_0^{1/6}t^{1/2} \sum_{d_1,d_2=3}^{4t_1}  (2t_0^{1/6}t^{1/2})^{d_1+d_2-2}  \Big(\frac{2t_0+3}{2t+1}\Big)^{(d_1+d_2)/2} e^{2t_1} \\ 
              &\leq  t_0^{1/6}t^{1/2} e^{2t_1} t_1^2 (2t_0^{1/6}t^{1/2})^{8t_1-2} \Big(\frac{2t_0+3}{2t+1}\Big)^{4t_1} \\
              &= t_0^{4t_1/3+1/6}t^{-1/2} e^{2t_1} t_1^2 2^{8t_1} (2t_0+3)^{4t_1} \Big( \frac{t}{2t+1} \Big)^{4t_1} \\ 
              &= o\Big( \frac{f(t)}{t_1^2\sqrt{t}} \Big)
\end{align*} 

\noindent Note that we omitted the tedious calculation justifying the transition from (2) to (3) 
since calculating the upper bound of the joint probability distribution is very similar to the calculation of Lemma~\ref{lem:lemma2}.

\noindent \\
\underline{$\bullet$ {\sc Case 2} $(s_1,s_2) \in E(G_t)$ :}\\

Notice that in any case $(s_1,s_2)$ share at most two faces (which may change over time). 
Note that the two connected vertices $s_1,s_2$ share a common face only if $s_1,s_2 \in \{1,2,3\}$\footnote{We analyze the case where $s_1,s_2 \geq 4$. The other case is treated in the same manner.}.
Consider the following modified process $\mathcal{Y}'$: whenever an incoming vertex ``picks'' one of the two common faces 
we don't insert it. We choose two other faces which are not common to $s_1,s_2$ and add one vertex in each of those. 
Notice that the number of faces increases by 1 for both $s_1,s_2$ as in the original process and the difference of the degrees
remains the same. An algebraic manipulation similar to Case 1 gives the desired result.
\end{proof}

\section{Largest Eigenvalues of the Adjacency Matrix: Proof of Theorem~\ref{thrm:thrm2}} 
\label{sec:eigen} 

In complex network analysis the spectrum of the adjacency matrix 
is an important aspect of the network or --more generally--  of the network model 
with several applications, e.g., triangle counting \cite{tsourakakis}.
To the best of our knowledge the spectrum of Random Apollonian Networks
has been studied only experimentally \cite{adrade2}. 
Here, we prove Theorem~\ref{thrm:thrm2}. We decompose the proof of the main
theorem in Lemmas~\ref{lem:lemma7},~\ref{lem:lemma8},~\ref{lem:lemma9},~\ref{lem:lemma10}.

Having computed the highest degrees of a RAN in Section~\ref{sec:degree}, eigenvalues are computed 
by adapting existing techniques \cite{mihail,chung}, see also \cite{flaxman} for a closely
related analysis, and taking into account the special properties of the model. 
We decompose the proof of Theorem~\ref{thrm:thrm2} in four lemmas. 
Specifically, in Lemmas~\ref{lem:lemma7},~\ref{lem:lemma8} we bound the degrees and co-degrees respectively. 
Having these bounds, we decompose the graph into a star forest and show in Lemmas~\ref{lem:lemma9} and~\ref{lem:lemma10} that 
its largest eigenvalues, which are $(1\pm o(1))\sqrt{\Delta_i}$, dominate the eigenvalues of the remaining graph.  

We partition the vertices into three set $S_1,S_2,S_3$. Specifically, let $S_i$ be the set of vertices 
added after time $t_{i-1}$ and at or before time $t_i$ where 

$$ t_0=0, t_1=t^{1/8}, t_2=t^{9/16}, t_3=t.$$

In the following we use the recursive variational characterization of eigenvalues \cite{chungbook}. Specifically,
let $A_G$ denote the adjacency matrix of a simple, undirected graph $G$ and let $\lambda_i(G)$ denote
the $i$-th largest eigenvalue of $A_G$. Then

$$ \lambda_i(G) = \min_{S} \max_{x \in S,x \neq 0} \frac{x^TA_Gx}{x^Tx}$$ 

\noindent where $S$ ranges over all $(n-i+1)$ dimensional subspaces of $\field{R}^n$. 

We shall use the following lemma in our proof, specifically in the proof of Lemma~\ref{lem:lemma9}.

\begin{lemma}
For any $\epsilon>0$ and any $f(t)$ with $f(t) \rightarrow +\infty$ as $t \rightarrow +\infty$
the following holds \whp: for all $s$ with $f(t) \leq s \leq t$, for all vertices $r \leq s$,
then $d_s(r) \leq s^{\tfrac{1}{2}+\epsilon}r^{-\tfrac{1}{2}}$. 
\label{lem:lemma7}
\end{lemma}

\begin{proof} 

Set $q=\Ceil{\frac{4}{\epsilon}}$. We use Lemma~\ref{lem:lemma1}, a simple union bound and Markov's inequality to obtain:

\begin{align*}
\Prob{ \bigcup_{s=f(t)}^t \cup_{r=1}^s \{ d_s(r) \geq s^{1/2+\epsilon} r^{-1/2} \} } &\leq \sum_{s=f(t)}^t \sum_{r=1}^s  \Prob{ d_s(r)^{(q)} \geq (s^{1/2+\epsilon} r^{-1/2})^{(q)}  }\\
&\leq \sum_{s=f(t)}^t \sum_{r=1}^s  \Prob{ d_s(r)^{(q)} \geq (s^{-(q/2+q\epsilon)} r^{q/2}) } \\ 
&\leq \sum_{s=f(t)}^t \sum_{r=1}^s  \frac{(q+2)!}{2} \Big( \frac{2s}{r} \Big)^{q/2} s^{-q/2} s^{-q\epsilon} r^{q/2}  \\
&= \frac{(q+2)!}{2} 2^{q/2} \sum_{s=f(t)}^t s^{1-q\epsilon} \\
&\leq \frac{(q+2)!}{2} 2^{q/2}   \int_{f(t)-1}^t \! x^{1-q\epsilon}  \, \mathrm{d}x    \\
&\leq \frac{(q+2)!}{2(q\epsilon-2)} 2^{q/2} (f(t)-1)^{2-q\epsilon} = o(1).
\end{align*}

\end{proof}

\begin{lemma}
Let $S_3'$ be the set of vertices in $S_3$ which are adjacent to more than one vertex of $S_1$. Then $ |S_3'| \leq t^{1/6}$ \whp.
\label{lem:lemma8} 
\end{lemma}

\begin{proof} 
First, observe that when vertex $s$ is inserted it becomes adjacent to more than one vertex of $S_1$ 
if the face chosen by $s$ has at least two vertices in $S_1$. We call the latter property $\mathcal{A}$
and we write $s \in \mathcal{A}$ when $s$ satisfies it. 
At time $t_1$ there exist $2t_1+1$ faces total, which consist of faces whose three vertices 
are all from $S_1$. At time $s \geq t_2$ there can be at most $6t_1+3$ faces with at least two vertices
in $S_1$ since each of the original $2t_1+1$ faces can give rise to at most 3 new faces with 
at least two vertices in $s_1$.
Consider a vertex $s \in S_3$, i.e., $s \geq t_2$. By the above argument, 
$ \Prob{|N(s) \cap S_1| \geq 2 } \leq \frac{6t_1+3}{2t+1}$. Writing $|S_3'|$ as a sum
of indicator variables, i.e., $|S_3'|= \sum_{s=t_2}^t I(s \in \mathcal{A})$
and taking the expectation we obtain

\begin{align*} 
\Mean{|S_3'|} &\leq \sum_{s=t_2}^t \frac{6t_1+3}{2t+1} \leq (6t_1+3) \int_{t_2}^t \! (2x+1)^{-1} \, \mathrm{d}x  \\ 
         &\leq (3t^{\tfrac{1}{8}}+\tfrac{3}{2})   \ln{\frac{2t+1}{2t_2+1}}= o(t^{1/7})
\end{align*}

\noindent By Markov's inequality:

$$ \Prob{ |S_3'| \geq t^{1/6} } \leq \frac{ \Mean{|S_3'|} }{t^{1/6}}  = o(1).$$ 

\noindent Therefore, we conclude that $|S_3'| \leq t^{1/6}$ \whp.
\end{proof}

\begin{lemma}
Let $F \subseteq G$ be the star forest consisting of edges between $S_1$ and $S_3 - S_3'$. 
Let $\Delta_1 \geq \Delta_2 \geq \ldots \geq \Delta_k$ denote the $k$ highest degrees of $G$. 
Then $\lambda_i(F) = (1-o(1))\sqrt{\Delta_i}$ \whp. 
\label{lem:lemma9} 
\end{lemma}

\begin{proof}

It suffices to show that $\Delta_i(F) = (1-o(1))\Delta_i(G)$ for $i=1,\ldots,k$. 
Note that since the $k$ highest vertices are inserted before $t_1$ \whp, the edges they lose
are the edges between $S_1$ and the ones incident to $S_3'$ and $S_2$ 
and we know how to bound the cardinalities of all these sets. 
Specifically by Lemma~\ref{lem:lemma8} $|S_3'| \leq t^{1/6}$ \whp and by 
Theorem~\ref{thrm:thrm1} the maximum degree in $G_{t_1},G_{t_2}$ is less 
than $t_1^{1/2+\epsilon_1} = t^{1/8}$, $t_2^{1/2+\epsilon_2} = t^{5/16}$ for $\epsilon_1=1/16,\epsilon_2=1/32$
respectively \whp. Also by Theorem~\ref{thrm:thrm1} $\Delta_i(G) \geq \frac{\sqrt{t}}{\log{t}}$.
Hence, we obtain

$$ \Delta_i(F) \geq \Delta_i(G) - t^{1/8} - t^{5/16} - t^{1/6} = (1-o(1))\Delta_i(G).$$

\end{proof}

\noindent To complete the proof of Theorem~\ref{thrm:thrm2} it suffices to prove that $\lambda_1(H)$ is $o(\lambda_k(F))$
where $H=G-F$. We prove this in the following lemma. The proof is based on bounding maximum
degree of appropriately defined subgraphs using Lemma~\ref{lem:lemma7} and then using standard inequalities
from Spectral Graph Theory \cite{chungbook}. 

\begin{lemma} 
$\lambda_1(H) = o(t^{1/4})$ \whp.
\label{lem:lemma10}
\end{lemma}

\begin{proof} 
From Gershgorin's theorem \cite{strang} the maximum eigenvalue of any graph is bounded by the maximum degree. 
We bound the eigenvalues of $H$ by bounding the maximum eigenvalues of six different induced
subgraphs. Specifically, let $H_i=H[S_i]$, $H_{ij}=H(S_i,S_j)$ where $H[S]$ is the subgraph
induced by the vertex set $S$ and $H(S,T)$ is the subgraph containing only edges with one vertex
is $S$ and other in $T$. We use Lemma~\ref{lem:lemma9} to bound 
$\lambda_1(H(S_1,S_3))$ and Lemma~\ref{lem:lemma8} for the other eigenvalues. 
We set $\epsilon=1/64$.

$$ \lambda_1(H_1) \leq \Delta_1(H_1)      \leq t_1^{1/2+\epsilon} = t^{33/512}. $$
$$ \lambda_1(H_2) \leq \Delta_1(H_2)    \leq t_2^{1/2+\epsilon}t_1^{-1/2} = t^{233/1024}. $$
$$ \lambda_1(H_3) \leq \Delta_1(H_3)    \leq t_3^{1/2+\epsilon}t_2^{-1/2} = t^{15/64}. $$
$$ \lambda_1(H_{12}) \leq \Delta_1(H_{12}) \leq t_2^{1/2+\epsilon} = t^{297/1024}. $$
$$ \lambda_1(H_{23}) \leq \Delta_1(H_{23}) \leq t_3^{1/2+\epsilon} t_1^{-1/2} = t^{29/64}. $$
$$ \lambda_1(H_{13}) \leq \Delta_1(H_{13}) \leq t^{1/6}. $$

Therefore \whp we obtain

$$ \lambda_1(H) \leq \sum_{i=1}^3 \lambda_1(H_i) + \sum_{i<j} \lambda_1(H_{i,j}) = o(t^{1/4}).$$
\end{proof}

\section{Diameter}
\label{sec:diameter} 

As shown in Appendix B of \cite{zhou} using physicist's methodology the diameter of a RAN is asymptotically upper bounded
by a logarithmic factor $\log{(t)}$. Similarly, Zhang et al. \cite{zhang} show that the diameter scales logarithmically. 
In the following we give a simple proof that the diameter $d(G_t)$  of the graph $G_t$ created  $t$ steps is $O(\log{t})$ \whp
and in Theorem~\ref{thrm:ternary} we give a refined upper bound for the diameter.
 
\begin{figure}[h]
\centering
\includegraphics[width=0.3\textwidth]{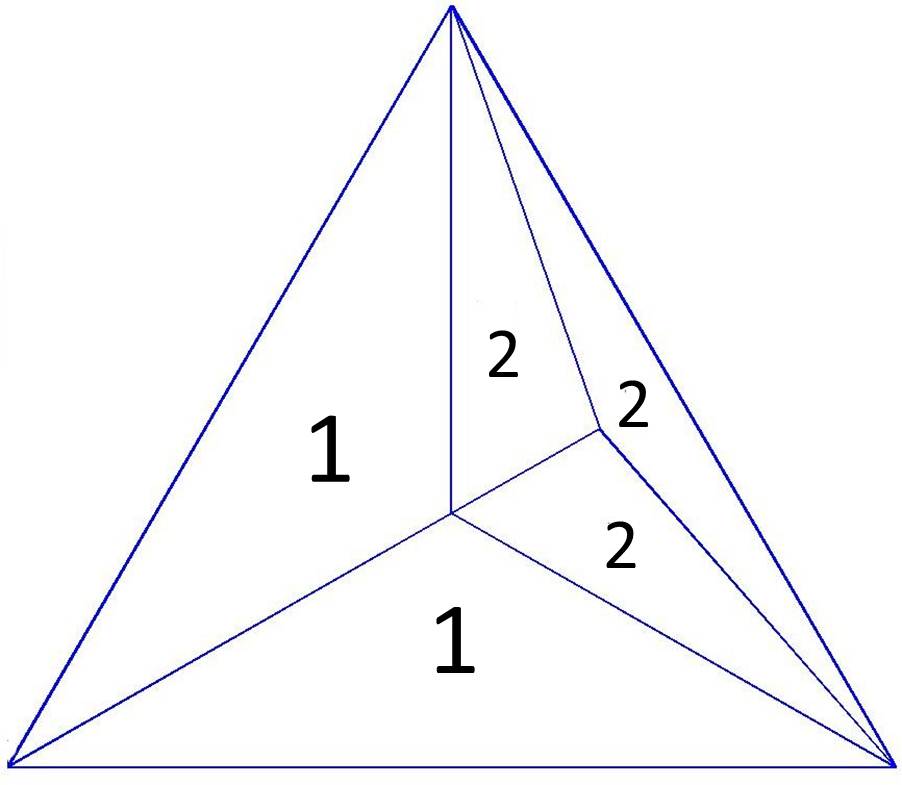}
\caption{An instance of the process for $t=2$. Each face is labelled with its depth.}
\label{fig:depth}
\end{figure}

We begin with a necessary definition for the proof of Claim~\ref{claim:diameter}. 
We define the {\it depth of a face} recursively. Initially we have one face $\alpha$ 
its depth is $depth(\alpha)=1$. For each new face $\beta$ created 
by picking a face $\gamma$, we have $depth(\beta)=depth(\gamma)+1$. 
An example is shown in Figure~\ref{fig:depth}, where each face is labelled with its
corresponding depth. 

\begin{claim} 
The diameter $d(G_t)$  of the graph $G_t$ created  $t$ steps is $O(\log{t})$ \whp. 
\label{claim:diameter} 
\end{claim}

\begin{proof} 

Note that if $k^*$ is the maximum depth of a face then  $d(G_t)=O(k^*)$.
Hence, we need to upper bound the depth of a given face after $t$ rounds. 
Let $F_t(k)$ be the number of faces of depth $k$ at time $t$, then: 

\begin{align*}
\Mean{F_t(k)} &= \sum_{1 \leq t_1<t_2<\ldots<t_k\leq t} \prod_{j=1}^k \frac{1}{2t_j+1} \leq \frac{1}{k!} (\sum_{j=1}^t \frac{1}{2j+1} )^k \leq \frac{1}{k!} (\frac{1}{2} \log{t})^k \leq   (\frac{e\log{t}}{2k})^{k+1}
\end{align*}

By the first moment method we obtain $k^*=O(\log{t})$ \whp and hence $d(G_t)=O(\log{t})$ \whp. 
\end{proof}

\begin{figure}[h]
\centering
\includegraphics[width=0.5\textwidth]{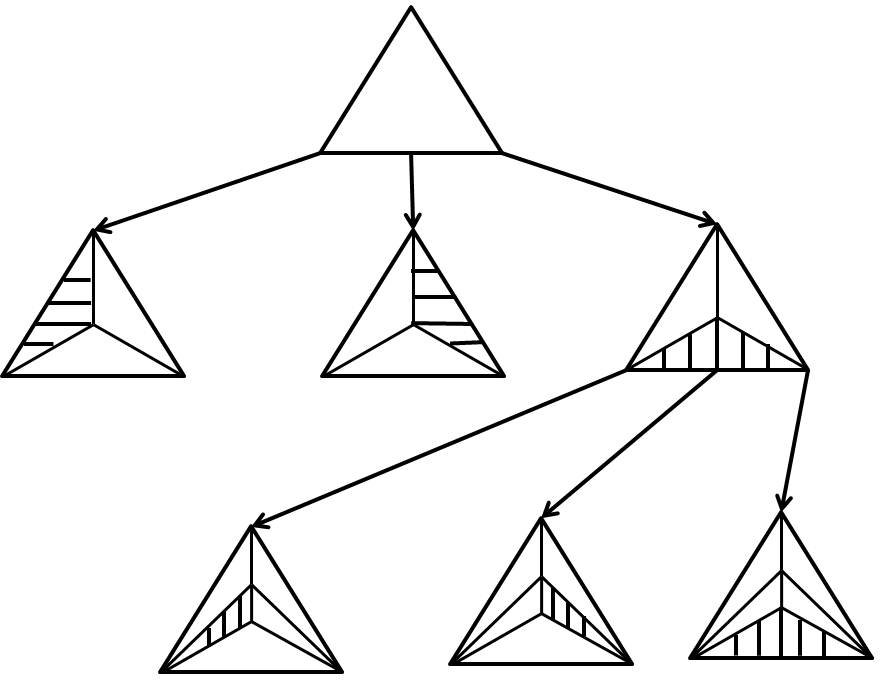}
\caption{RANs as random ternary trees.}
\label{fig:ternary}
\end{figure}

The depth of a face can be formalized via a bijection between random ternary trees and RANs.
Using this bijection we prove Theorem~\ref{thrm:ternary} which gives a refined upper bound on the 
asymptotic growth of the diameter.

\begin{proof} 
Consider the random process which starts with a single vertex tree and at every
step picks a random leaf and adds three children to it. Let $T$
be the resulting tree after $t$ steps. 
There exists a natural bijection between the RAN process and this process, see \cite{darase}
and also Figure~\ref{fig:ternary}.  
The depth of $T$ in probability 
is $\frac{\rho}{2} \log{t}$ where $\frac{1}{\rho}=\eta$ is the unique solution greater than 1 
of the equation $\eta - 1 - \log{\eta} = \log{3}$, see Broutin and Devroye \cite{devroye}, pp. 284-285.
Note that the diameter $d(G_t)$ is at most twice the height of the tree and hence the result follows. 
\end{proof} 	

The above observation, i.e., the bijection between RANs and random ternary trees
cannot be used to lower bound the diameter. A counterexample is shown in Figure~\ref{fig:counter}
where the height of the random ternary tree can be made arbitrarily large but the diameter is 2. 
Albenque and Marckert proved in \cite{albenque} that if $v,u$ are two i.i.d. uniformly random
internal vertices, i.e., $v,u \geq 4$, then the distance $d(u,v)$ tends
to $\frac{6}{11} \log{n}$ with probability 1 as the number of vertices $n$ of the RAN grows to infinity.
However, an exact expression of the asymptotic growth of the diameter to the best
of our knowledge remains an open problem.
Finally, it is worth mentioning that the diameter of the RAN grows faster
asymptotically than the diameter of the classic preferential attachment model \cite{albert} which \whp 
grows as $\frac{\log{t}}{\log{\log{t}}}$, see Bollob\'{a}s and Riordan \cite{bollobasriordan}.

\begin{figure}[h]
\centering
\includegraphics[width=0.5\textwidth]{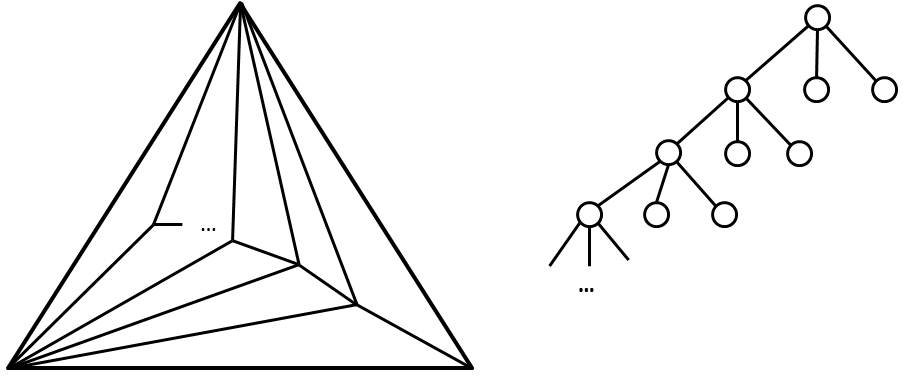}
\caption{The height of the random ternary tree cannot be used to lower bound the diameter. The height of the 
random ternary tree can be arbitrarily large but the diameter is 2.}
\label{fig:counter}
\end{figure}

\section{Waiting Times}
\label{sec:waiting} 

Consider the three faces created after the insertion of the first point.
Let's call the face which gets the first vertex $A$ and the other two faces $B,C$,
see also Figure~\ref{fig:fig1}(b). Also for simplicity of the notation
let the configuration of Figure~\ref{fig:fig1}(b) correspond to time $t=0$. 
Let $X$ equal the number of steps until a new vertex picks face  $B$ or $C$. 
Clearly, $X \in \{1,2,\ldots\}$. What is the expectation $\Mean{X}$? For any $t \geq 1$

\begin{align*}
\Prob{X > t} &= \prod_{j=1}^t \frac{3+2(j-1)}{5+2(j-1)} = \frac{3}{2t+3}.
\end{align*}

Using now the fact that $\Mean{X} = \sum_{t=1}^{+\infty} \Prob{X \geq t}= 1 + \sum_{t=1}^{+\infty} \Prob{X>t}$ 
and substituting we obtain that  $\Mean{X} =+\infty$.

\section{Conclusions}
\label{sec:concl}

In this work we studied several aspects of Random Apollonian Networks,
namely the highest degree vertices, the spectrum, the diameter and the ``waiting times''.
There are various research directions which we plan to address in future work.
Indicatively, we report the following: 
(a) In contrast to the preferential attachment model, RAN have small separators. 
Specifically, by the Lipton-Tarjan theorem \cite{lipton-tarjan} 
we know the existence of a good separator of size $O(\sqrt{n})$. 
This is in accordance with the empirical observation that real world networks have good 
separators \cite{blandford}. We plan to investigative the minimum size separator in future work. 
(b) Consider the natural process where an adversary with probability $\alpha$ 
at every step deletes a vertex from the network, e.g., uniformly at random
or in a maliscious way. Under which conditions do we obtain a connected RAN?
(c) What is the exact asymptotic expression for the diameter of $G_t$?

\section*{Acknowledgments}

We would like to thank Luc Devroye and Alexis Darrasse for pointing out references \cite{devroye}
and \cite{albenque,panholzer} respectively.

\section*{APPENDIX} 
\label{sec:appendix}

\subsection*{Proof of Claim~\ref{claim:lemma2}}

\begin{proof}

Let $\tau=(t_0 \equiv \tau_0,\underbrace{\tau_1,\ldots,\tau_r}_{\text{insertion times}},\tau_{r+1} \equiv t)$ be a vector denoting that $Y_t$ 
increases by 1 at $\tau_i$ for $i=1,\ldots,r$. 
We upper bound the probability $p_{\tau}$ of this event in the following.Note that 
we consider the case where the vertices have same degree as the number of faces around them,
Figure~\ref{fig:fig2}(b). The other case is analyzed in exactly the same way, modulo a negligible
error term.

\begin{align*}
p_{\tau} &= \Bigg[ \prod_{k=1}^r \frac{d_0+k-1}{2\tau_k+1} \Bigg] \Bigg[ \prod_{k=0}^r \prod_{j=\tau_k+1}^{\tau_{k+1}-1} \Big(1-\frac{d_0+k}{2j+1}\Big) \Bigg] \\ 
         &= d_0 (d_0+1) \ldots (d_0+r-1)\Bigg[ \prod_{k=1}^r \frac{1}{2\tau_k+1} \Bigg]
         \exp{\Bigg( \sum_{k=0}^r \sum_{j=\tau_k+1}^{\tau_{k+1}-1} \log{\Big(1-\frac{d_0+k}{2j+1}\Big)}  \Bigg)} \\
         &= \frac{(d_0+r-1)!}{(d_0-1)!}\Bigg[ \prod_{k=1}^r \frac{1}{2\tau_k+1} \Bigg]
         \exp{\Bigg( \sum_{k=0}^r \sum_{j=\tau_k+1}^{\tau_{k+1}-1} \log{\Big(1-\frac{d_0+k}{2j+1}\Big)}  \Bigg)} \\
\end{align*}

\noindent Consider now the inner sum which we upper bound using an integral: 

\begin{align*} 
\sum_{j=\tau_k+1}^{\tau_{k+1}-1} \log{\Big(1-\frac{d_0+k}{2j+1}\Big)}  &\leq  \int_{\tau_k+1}^{\tau_{k+1}} \! \log{\Big(1-\frac{d_0+k}{2x+1}\Big)}  \, \mathrm{d}x  \\
&\leq -\big( \tau_{k+1}+\tfrac{1}{2} \big)  \log{(2\tau_{k+1}+1)} + \\
& \frac{2\tau_{k+1}+1-(d_0+k)}{2} \log{(2\tau_{k+1}+1-(d_0+k))} + \\
& \big( \tau_{k}+\tfrac{3}{2} \big)  \log{(2\tau_{k}+3)} - \frac{2\tau_{k}+3-(d_0+k)}{2} \log{(2\tau_{k}+3-(d_0+k))} 
\end{align*}

\noindent since 

$$ \int \log{\Big(1-\frac{d_0+k}{2x+1}\Big)} = -\big( x+\tfrac{1}{2} \big)  \log{(2x+1)} + \frac{2x+1-(d_0+k)}{2} \log{(2x+1-(d_0+k))} $$

\noindent Hence we obtain $\sum_{k=0}^r \sum_{j=\tau_k+1}^{\tau_{k+1}-1} \log{\Big(1-\frac{d_0+k}{2j+1}\Big)}  \leq A+\sum_{k=1}^r B_k$
where 

\begin{align*} 
A &= \big( \tau_{0}+\tfrac{3}{2} \big)  \log{(2\tau_{0}+3)} - \frac{2\tau_{0}+3-d_0}{2} \log{(2\tau_{0}+3-d_0)} \\
  & -\big( \tau_{r+1}+\tfrac{1}{2} \big)  \log{(2\tau_{r+1}+1)} + \frac{2\tau_{r+1}+1-(d_0+r)}{2} \log{(2\tau_{r+1}+1-(d_0+r))}
\end{align*}

\noindent and 

\begin{align*} 
B_k &= \big( \tau_{k}+\tfrac{3}{2} \big)  \log{(2\tau_{k}+3)} - \frac{2\tau_{k}+3-(d_0+k)}{2} \log{(2\tau_{k}+3-(d_0+k))} \\
    & -\big( \tau_{k}+\tfrac{1}{2} \big)  \log{(2\tau_{k}+1)} + \frac{2\tau_{k}+1-(d_0+k-1)}{2} \log{(2\tau_{k}+1-(d_0+k-1))}.
\end{align*} 

We first upper bound the quantities $B_k$ for $k=1,\ldots,r$. By rearranging terms and using the identity $\log{(1+x)} \leq x$ 
we obtain 

\begin{align*} 
B_k &= \big( \tau_{k}+\tfrac{1}{2} \big)  \log{\big(1+\frac{1}{\tau_{k}+\tfrac{1}{2}}\big)} + \log{(2\tau_k+3)} \\
    &  -\frac{1}{2} \log{\big(2\tau_k+3-(d_0+k)\big)} - \frac{2\tau_{k}+2-(d_0+k)}{2} \log{\Big( 1+\frac{1}{2\tau_{k}+2-(d_0+k)} \Big)}.\\
    &\leq \frac{1}{2} + \frac{1}{2} \log{\big(2\tau_k+3\big)} -\frac{1}{2} \log{\big( 1-\frac{d_0+k}{2\tau_{k}+3} \big)}
\end{align*} 

First we rearrange terms and then we bound the term $e^A$ by using the inequality $e^{-x-x^2/2} \geq 1-x$ which is valid for $0<x<1$:

\begin{align*} 
A   &= -\big( \tau_{0}+\tfrac{3}{2} \big)  \log{\big(1-\frac{d_0}{2\tau_{0}+3}\big)} + \big( \tau_{r+1}+\tfrac{1}{2} \big)  \log{\big(1-\frac{d_0+r}{2\tau_{r+1}+1}\big)} +\frac{d_0}{2} \log{\big( 2\tau_0 +3 -d_0\big)} \\
    &    - \frac{d_0+r}{2} \log{\Big( 2\tau_{r+1}+1-(d_0+r) \Big)}.\Rightarrow \\
e^A &= \big(1-\frac{d_0}{2\tau_{0}+3}\big)^{-(\tau_{0}+\tfrac{3}{2}) } \big(1-\frac{d_0+r}{2\tau_{r+1}+1}\big)^{\tau_{r+1}+\tfrac{1}{2}}( 2\tau_0 +3 -d_0)^{\frac{d_0}{2}}( 2\tau_{r+1}+1-(d_0+r))^{- \frac{d_0+r}{2}} \\
    &= \bigg( \frac{2\tau_0+3}{2\tau_{r+1}+1} \bigg)^{d_0/2} (2\tau_{r+1}+1)^{-r/2} \bigg( 1-\frac{d_0}{2\tau_0+3}\bigg)^{-(\tau_{0}+\tfrac{3}{2})+\tfrac{d_0}{2}} \bigg( 1-\frac{d_0+r}{2\tau_{r+1}+1} \bigg)^{\tau_{r+1}+\tfrac{1}{2} -\frac{d_0+r}{2}} \\
    &\leq \bigg( \frac{2t_0+3}{2t+1} \bigg)^{d_0/2} (2t+1)^{-r/2} \bigg( 1-\frac{d_0}{2\tau_0+3}\bigg)^{-(\tau_{0}+\tfrac{3}{2})+\tfrac{d_0}{2}} e^{ \Big( -\frac{d_0+r}{2t+1}- \big(-\frac{d_0+r}{2t+1}\big)^2/2 \Big) (t+1/2-\frac{d_0+r}{2})  } \\
    &= \bigg( \frac{2t_0+3}{2t+1} \bigg)^{d_0/2} (2t+1)^{-r/2} \bigg( 1-\frac{d_0}{2\tau_0+3}\bigg)^{-(\tau_{0}+\tfrac{3}{2})+\tfrac{d_0}{2}} e^{ -\frac{d_0+r}{2}+\frac{(d_0+r)^2}{8t+4}+\frac{(d_0+r)^3}{4(2t+1)^2} }
\end{align*} 

\noindent Now we upper bound the term $\exp{\big( A+\sum_{k=1}^r B_k \big)}$ using the above upper bounds:

\begin{align*} 
e^{A+\sum_{k=1}^r B_k} &\leq e^A e^{r/2} \prod_{i=1}^r \sqrt{\frac{2\tau_k+3}{1-\frac{d_0+k}{2\tau_k+3}}}\\
&\leq  \bigg( 1-\frac{d_0}{2\tau_0+3}\bigg)^{-(\tau_{0}+\tfrac{3}{2})+\tfrac{d_0}{2}}  e^{ -\frac{d_0}{2}+\frac{(d_0+r)^2}{8t+4}+\frac{(d_0+r)^3}{4(2t+1)^2} } \bigg( \frac{2t_0+3}{2t+1} \bigg)^{d_0/2} \times  \\
& (2t+1)^{-r/2} \prod_{i=1}^r \sqrt{\frac{2\tau_k+3}{1-\frac{d_0+k}{2\tau_k+3}}} 
\end{align*}

\noindent  Using the above upper bound we get that

\begin{align*}
p_{\tau} &\leq  C(r,d_0,t_0,t) \prod_{k=1}^r  \Big[ (2\tau_k+3-(d_0+k))^{-1/2} \big( 1+\frac{1}{\tau_k+1/2}\big)\Big]
\end{align*}

\noindent where 

$$ C(r,d_0,t_0,t)=\frac{(d_0+r-1)!}{(d_0-1)!}\bigg( 1-\frac{d_0}{2\tau_0+3}\bigg)^{-(\tau_{0}+\tfrac{3}{2})+\tfrac{d_0}{2}}  e^{ -\frac{d_0}{2}+\frac{(d_0+r)^2}{8t+4}+\frac{(d_0+r)^3}{4(2t+1)^2} } \bigg( \frac{2t_0+3}{2t+1} \bigg)^{d_0/2}(2t+1)^{-r/2}$$

We need to sum over all possible insertion times to bound the probability of interest $p^*$. We set $\tau_k' \leftarrow \tau_k -  \lceil \frac{d_0+k}{2} \rceil$ for $k=1,\ldots,r$. For $d=o(\sqrt{t})$ and $r=o(t^{2/3})$ we obtain: 

\begin{align*} 
p^* &\leq C(r,d_0,t_0,t) \sum_{t_0+1 \leq \tau_1 < .. <\tau_r \leq t}  \prod_{k=1}^r  \Big[ (2\tau_k+3-(d_0+k))^{-1/2} \big( 1+\frac{1}{\tau_k+1/2}\big)\Big] \\
  &\leq C(r,d_0,t_0,t) \sum_{t_0-\lceil \tfrac{d_0}{2} \rceil +1 \leq \tau_1' \leq .. \leq \tau_r' \leq t-\lceil \tfrac{d_0+r}{2} \rceil}  \prod_{k=1}^r  \Big[ (2\tau_k'+3)^{-1/2} \big( 1+\frac{1}{\tau_k'+\frac{d_0+k}{2} +1/2}\big)\Big] \\
  &\leq \frac{C(r,d_0,t_0,t)}{r!} \Bigg( \sum_{t_0-\lceil \tfrac{d_0}{2} \rceil}^{t-\lceil \tfrac{d_0+r}{2} \rceil} (2\tau_k'+3)^{-1/2}+\tfrac{1}{\sqrt{2}} (\tau_k'+3/2)^{-3/2} \Bigg)^r \\
  &\leq \frac{C(r,d_0,t_0,t)}{r!}  \Bigg( \int_{0}^{t-\frac{d+r}{2}} \! \Big[  (2x+3)^{-1/2} + \tfrac{1}{\sqrt{2}} (x+3/2)^{-3/2} \Big] \, \mathrm{d}x\Bigg)^r \\ 
  &\leq \frac{C(r,d_0,t_0,t)}{r!} \Big(\sqrt{2t+3-(d_0+r)}+2/3\Big)^{r}  \\ 
  &\leq \frac{C(r,d_0,t_0,t)}{r!} (2t)^{r/2} e^{-\tfrac{r}{2} \tfrac{d_0+r-3}{2t} } e^{\frac{2r}{3\sqrt{2t-(d_0+r)+3}}} \\
  &\leq {d_0+r-1 \choose d_0-1} \big(\frac{2t_0+3}{2t+1}\big)^{d_0/2} \Big[ \big( 1-\frac{d_0}{2t_0+3}\big)^{-(1-\frac{d_0}{2t_0+3})} \Big]^{t_0+3/2} \times \\ 
  & \big(\frac{2t}{2t+1}\big)^{r/2} \exp{\Big( -\frac{d_0}{2} + \frac{(d_0+r)^2}{8t+4}+ \frac{(d_0+r)^3}{4(2t+1)^2} - \frac{r(d_0+r-3)}{4t}+\frac{2r}{3\sqrt{2t+3-(d_0+r)}}  \Big)} \\ 
\end{align*}

\noindent By removing the $o(1)$ terms in the exponential and using the fact that $x^{-x}\leq e$ 
we obtain the following bound on the probability $p^*$.

$$  p^* \leq  {d_0+r-1 \choose d_0-1} \Big(\frac{2t_0+3}{2t+1}\Big)^{d_0/2} e^{\frac{3}{2}+t_0-\frac{d_0}{2}+\frac{2r}{3\sqrt{t}}} $$
\end{proof}

\end{document}